\documentclass[journal,twoside,web]{ieeecolor}
\usepackage{cite}
\usepackage{generic}
\usepackage{amsmath,amssymb,amsfonts}

\usepackage{amsthm}

\usepackage{algorithmic}
\usepackage{graphicx}
\usepackage{algorithm}
\usepackage{hyperref}
\hypersetup{hidelinks}
\usepackage{textcomp}
\usepackage{subfigure}
\let\labelindent\relax
\usepackage{enumitem}
\usepackage{mathrsfs}

\def\BibTeX{{\rm B\kern-.05em{\sc i\kern-.025em b}\kern-.08em
    T\kern-.1667em\lower.7ex\hbox{E}\kern-.125emX}}
\newtheorem{lemma}{Lemma}
\newtheorem{theorem}{Theorem}

\newtheorem{assumption}{Assumption}

\newcommand{\di}{D_i}
\newcommand{\xdi}{x_{\di}}
\newcommand{\ydi}{y_{\di}}
\newcommand{\vdi}{v_{\di}}
\newcommand{\vdib}{\bar v_{\di}}
\newcommand{\tdi}{\theta_{\di}}
\newcommand{\aj}{A_j}
\newcommand{\xaj}{x_{\aj}}
\newcommand{\yaj}{y_{\aj}}
\newcommand{\vaj}{v_{\aj}}
\newcommand{\vajb}{\bar v_{\aj}}
\newcommand{\taj}{\theta_{\aj}}

\begin{document}
\title{Multiplayer Reach-Avoid Differential Games with Defender-Side Information Delay}
\author{Zehua Zhao, Rui Yan, Jianping He, Xiaoming Duan
\thanks{Zehua Zhao, Jianping He, and Xiaoming Duan are with the State Key Laboratory of Submarine Geoscience, School of Automation and Intelligent Sensing, Shanghai Jiao Tong University, Shanghai 200240, China (e-mails: \{zehua.zhao, jphe, xduan\}@sjtu.edu.cn).}
\thanks{Rui Yan is with the School of Artificial Intelligence, Beihang University, Beijing 100191, China (e-mail: rui\_yan@buaa.edu.cn).}
}

\maketitle

\begin{abstract}
We consider a class of pursuit-evasion games in which multiple defenders and attackers move in the plane with bounded speeds, while each defender observes the states of other agents with a constant time delay. For the one-attacker-one-defender case, we derive an explicit analytical characterization of the attacker’s delayed attack region and prove its convexity under mild assumptions. When the defender can guarantee capture, we formulate a convex optimization problem to compute the capture point and derive optimal strategies for both players. These strategies are shown to constitute a subgame-perfect Nash equilibrium by exploiting the sequential structure induced by the information delay. The analysis is further extended to the one-attacker-multiple-defender scenario and to the general multiplayer setting. In the latter case, delay-aware pairwise winning relations are incorporated into a maximum matching formulation to address the defender-attacker assignment. Numerical simulations for one-on-one, one-vs-multiple, and multi-agent cases validate the theoretical results and illustrate the impact of information delay on game outcomes and optimal strategies.
\end{abstract}


\section{Introduction}
Differential games characterize adversarial dynamic processes in continuous time and enable derivation of optimal strategies via the Hamilton-Jacobi-Isaacs (HJI) equation~\cite{IR-99}, providing mathematically interpretable solutions and structured strategies. Reach-avoid games, first proposed by Isaacs in~\cite{IR-99}, involve players with conflicting objectives  engaged in attack-defense interactions centered on a target region. Due to their wide range of applications in aerospace, maritime, and military domains, they have been extensively studied~\cite{xu2026collision,bui2025reach}.

Within the pioneering framework established by Isaacs, solving a differential game fundamentally relies on solving the HJI equation, a nonlinear partial differential equation that characterizes the value function of the game. However, for complex and high-dimensional problems, obtaining explicit solutions to the HJI equation is often extremely challenging. To overcome this difficulty, subsequent research has proposed various alternative approaches for addressing differential games, including analytical methods based on the Pontryagin Maximum Principle~\cite{pontryagin1966theory} and related optimal control theory. In recent years, one approach to solving such differential equations is geometric in nature and has attracted increasing attention due to its structural intuition and analytical simplicity~\cite{yan2017escape,garcia2020multiple,fu2020guarding,ramana2017pursuit,yan2024multiplayer,yan2019task}. This class of methods typically begins by constructing the barrier surface of the game, which partitions the state space into regions reflecting the relative advantage of each player. Candidate strategies are then designed based on this partition, and their optimality is established by verifying that the constructed value function satisfies the HJI equation~\cite{garcia2017geometric,garcia2018design,G-20,li2023optimal,li2026analytical}. Although directly solving the HJI equation remains difficult, verifying whether a constructed value function satisfies the HJI condition is comparatively more tractable. As a result, this verification-based analytical framework has become a common method for solving certain classes of differential games with relatively simple models.

Benefiting from the development of analytical methods, research on reach-avoid games has advanced rapidly in recent years~\cite{von2020multiple,lee2022two,liang2026pursuit,wang2025optimal,lee2024solutions,yan2023homicidal}. Both the modeling frameworks and the environmental assumptions have become increasingly sophisticated, enabling these games to better capture the characteristics of real-world physical systems and application scenarios. For example, the classical reach-avoid game formulated in a two-dimensional plane has been extended to three-dimensional spaces to better represent the motion of aerial vehicles and other agents operating in realistic environments~\cite{yan2019construction}. In addition, capture radii have been introduced for pursuers to reflect practical attack ranges, thereby making the interaction model more consistent with real operational constraints~\cite{yan2022matching}. Another important line of development concerns the number of players in the game. The original one-on-one formulation has been generalized to multi-agent scenarios involving multiple attackers and multiple defenders, which allows the framework to model cooperative and competitive behaviors in swarm or team-based systems~\cite{yan2018reach}. At the same time, the dynamical models of the agents have also evolved: instead of simple first-order integrator dynamics, many studies now adopt second-order integrator models with damping, which better approximate the dynamics of real physical platforms such as unmanned aerial vehicles or robotic systems~\cite{lyu2025reach}. Furthermore, researchers have incorporated environmental constraints to enhance realism. Obstacles are often introduced into the state space to capture the need for collision avoidance and safe navigation in practical environments~\cite{yan2024pursuit}. In addition, some studies place the entire reach-avoid game within fluid environments with constant flow fields, allowing the framework to represent situations such as aerial or underwater vehicles operating under environmental disturbances like wind or ocean currents~\cite{deng2023multiple}. Through these developments, the reach-avoid game framework has gradually evolved from an abstract theoretical model into a powerful analytical tool capable of describing increasingly realistic multi-agent adversarial scenarios.

In practical applications, especially in military contexts, information delays commonly arise during adversarial engagements. Therefore, incorporating information delays into reach–avoid games is both meaningful and practically relevant. In~\cite{li2026geometric}, Li et al. studied a multi-defender reach–avoid game in which the attacker is subject to information delay. For instance, a defender’s system may operate under standardized communication protocols that prescribe data acquisition and transmission rates, buffer capacities, and processing procedures, thereby inducing an inherent minimum observation delay~\cite{mi2021optimal, deka2021towards}. In addition, the clock frequency and interrupt-handling latency of the defender’s hardware are typically fixed. Under such conditions, real-time measurement of the defender’s latency by the attacker may not be necessary. Instead, knowledge of the standardized communication hardware, software, or off-the-shelf components employed by the defender is sufficient to infer the processing speed and intrinsic delay~\cite{winderix2021compiler, cheshmikhani2025interrupt}. Alternatively, these delay parameters may be estimated through prior reconnaissance, intelligence gathering, or pre-infiltration activities~\cite{hutchins2011intelligence}. In short-range dynamic confrontations, information delay is often unavoidable and must be addressed at the system design level~\cite{dennehy2011summary}. The concept of a defender operating under information delay was first introduced by Isaacs in~\cite{IR-99}. Under this premise, the manner in which the attacker exploits informational advantages to improve its outcome, and how the defender mitigates performance degradation despite known informational disadvantages, become central issues. The characterization of equilibrium strategies under such asymmetric timing conditions is therefore of fundamental importance. To date, however, no existing work has provided optimal strategies for the reach-avoid game with information delay. The main contributions of this paper are summarized as follows:

\begin{enumerate}[label=(\arabic*)]
    \item We propose a mathematical formulation of the reach-avoid game in which the defenders have delayed information about the attackers' positions, i.e., the defenders only have access to attackers' past states rather than their current positions.
    \item Under this formulation, we derive the strategies for both players in the one-attacker-one-defender case and prove that the proposed strategies are optimal in the sense of subgame-perfect Nash equilibrium.
    \item We extend the framework to the one-attacker-multiple-defender case and further to the multi-attacker-multi-defender scenario, and provide the corresponding strategies for each setting.
\end{enumerate}

The rest of this article is organized as follows. Section~\ref{PF} presents the mathematical formulation of the problem. Section~\ref{1vs1} derives the strategies for the one-attacker-one-defender case and proves that the proposed strategies constitute a subgame-perfect Nash equilibrium. Section~\ref{multi} extends the framework to the one-attacker-multiple-defender case as well as the multi-attacker-multi-defender scenario, and provides the corresponding strategies for each setting. Section~\ref{simu} presents numerical simulation results for all the aforementioned cases. Section~\ref{cc} concludes this article.

\section{Problem Formulation}\label{PF}

\begin{figure}[!t]
\centerline{\includegraphics[width=\columnwidth]{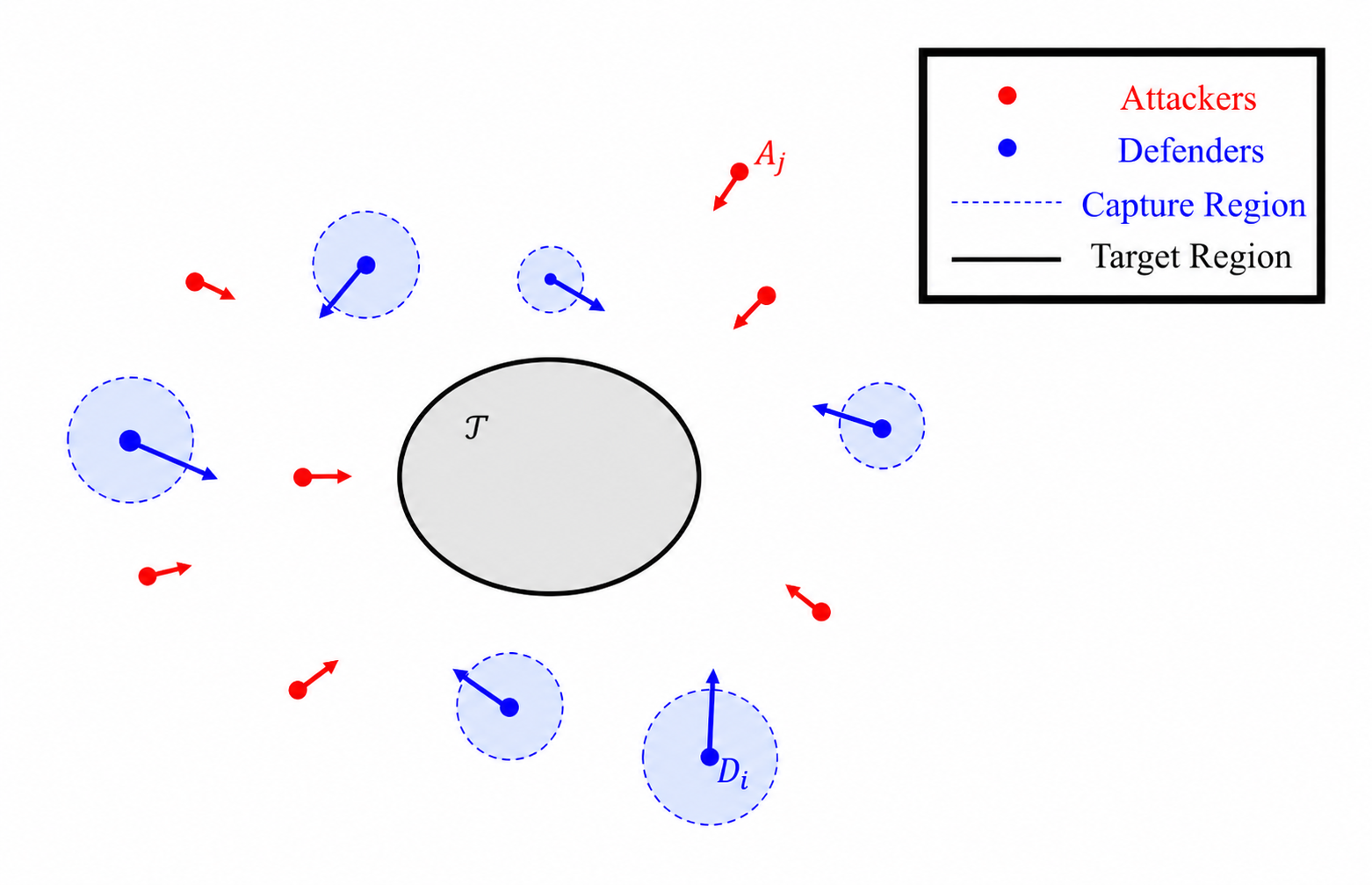}}
\caption{Schematic illustration of the considered scenario. The red dots represent the attackers, and the red arrows indicate the attackers' velocity. The blue dots represent the defenders, and the blue arrows indicate the defenders' velocity. The blue dashed circles denote the defenders' capture regions, and the region enclosed by the black line represents the target region.}
\label{fig2_1}
\end{figure}

We consider $N$ defenders $\mathscr D=\{D_1,D_2,\dots,D_{N}\}$ and $M$ attackers $\mathscr A=\{A_1,A_2,\dots,A_{M}\}$ driven by single integrators on a two-dimensional plane, as shown in Fig.~\ref{fig2_1}, and their dynamics are given by
\begin{equation*}
    \begin{aligned}
        \dot x_{\di}(t)=\vdi(t)\cos\tdi(t),\quad\dot y_{\di}(t)=&\vdi(t)\sin\tdi(t),\\&\quad i=1,2,\dots,N,
        \\\dot x_{\aj}(t)=\vaj(t)\cos\taj(t),\quad\dot y_{\aj}(t)=&\vaj(t)\sin\taj(t),\\&\quad j=1,2,\dots,M,
    \end{aligned}
\label{dyn}
\end{equation*}
where ${\mathbf x_{D_i}^t}=(\xdi^t,\ydi^t)^\top=(\xdi(t),\ydi(t))^\top$ and $\mathbf x_{A_j}^t=(\xaj^t,\yaj^t)^\top=(\xaj(t),\yaj(t))^\top$ are the positions of $\di$ and $\aj$ at time $t$, and ${\mathbf x_{D_i}^0}=(\xdi^0,\ydi^0)^\top=(\xdi(0),\ydi(0))^\top$ and $\mathbf x_{A_j}^0=(\xaj^0,\yaj^0)^\top=(\xaj(0),\yaj(0))^\top$ are the initial positions of $\di$ and $\aj$. We denote the system state at time $t$ by $\mathbf x^t=({\mathbf x_D^t}^\top,{\mathbf x_A^t}^\top)^\top=({\mathbf x_{D_1}^t}^\top,\dots,{\mathbf x_{D_N}^t}^\top,{\mathbf x_{A_1}^t}^\top,\dots,{\mathbf x_{A_M}^t}^\top)^\top$, where ${\mathbf x_{D_i}^t}$ and $\mathbf x_{A_j}^t$ are states of $D_i$ and $A_j$, respectively, and the initial state by $\mathbf x^0=({\mathbf x_D^0}^\top,{\mathbf x_A^0}^\top)^\top=({\mathbf x_{D_1}^0}^\top,\dots,{\mathbf x_{D_N}^0}^\top,{\mathbf x_{A_1}^0}^\top,\dots,{\mathbf x_{A_M}^0}^\top)^\top$. The control inputs are the magnitudes $\vdi$ and $\vaj$ and the directions $\tdi$ and $\taj$ of $\di$'s and $\aj$'s velocities respectively. The magnitudes of $\di$'s and $\aj$'s velocities are assumed to be bounded, i.e., $\vdi\in[0,\vdib]$, $\vaj\in[0,\vajb]$. There exists a capture radius $r_i$ for $\di$. Specifically, when the Euclidean distance between a given $\aj$ and $\di$ is less than $r_i$, $\aj$ is considered to be captured by $\di$.

Moreover, there exists a convex closed region in the two-dimensional plane, denoted by $\mathcal T\subset\mathbb R^2$, which serves as the target region of the game. Let
\begin{equation*}
    \mathcal T=\{\mathbf x\in\mathbb R^2\,|\,f(\mathbf x)\le0\},
\end{equation*}
where $f$ is a convex, continuous, and smooth function. The attackers aim to enter the target region $\mathcal T$; if they are captured before doing so, their objective is to minimize their minimum distance to the target region $\mathcal T$ over the entire course of the game. In contrast, the defenders seek to prevent the attackers from entering $\mathcal T$ and, upon capture, to keep them as far from $\mathcal T$ as possible. Accordingly, for attacker $\aj$, the game terminates when either it enters the target region $\mathcal T$, that is, $f(\mathbf x_{\aj}^t)\le0$, or it is captured by defender $\di$, that is, $\|\mathbf x_{\di}-\mathbf x_{\aj}\|\le r_i$. Therefore, the terminal time $t_f$ of the game is defined as the earliest time at which all attackers have either entered the target region or been captured by the defenders.

Unlike conventional reach-avoid games, the present game assumes that each defender $\di$ is associated with a constant information delay $\tau_i$. In particular, all information perceived by $\di$, except for its own state variables, is delayed by $\tau_i$. The value of $\tau_i$ is common knowledge among all attackers and defenders. For instance, at time $t>\tau_i$, defender $\di$ observes its own position as $\mathbf x_{\di}^t$, but the positions of $D_{i'}$ and $\aj$ are perceived  as $\mathbf x_{D_{i'}}^{t-\tau_i}$ and $\mathbf x_{\aj}^{t-\tau_i}$, respectively, with $i'\ne i$.

Assume that the initial positions and parameters of 
$\di$ and $\aj$ satisfy the following conditions. These conditions allow us to focus on the main cases of interest while avoiding certain technical complications.

\begin{assumption}{\rm(Initial Conditions and Parameter Assumptions)}
    The initial positions and parameters of $\di$ and $\aj$ satisfy the following conditions:
    \begin{enumerate}[label=(\arabic*)]
        \item \label{as1}During the time interval $[0, \tau_i]$, $\di$ remains stationary and does not take any active defensive action.
        \item \label{as2}$f(\mathbf x_{\aj}^0)>0$ for $j=1,2,\dots,M$;
        \item \label{as3}$\vdib>\vajb$ for $i=1,2,\dots,N$ and $j=1,2,\dots,M$;
        \item \label{as4}$\|\mathbf x_{\di}^0-\mathbf x_{\aj}^0\|>r_i+\vajb\tau_i$ for $i=1,2,\dots,N$ and $j=1,2,\dots,M$;
        \item \label{as5}$r_i>\vajb\tau_i$  for $i=1,2,\dots,N$ and $j=1,2,\dots,M$.
    \end{enumerate}
\label{as}
\end{assumption}
Due to the presence of the information delay $\tau_i$, defender $\di$ cannot access any information other than its own current state during the initial time interval $[0,\tau_i]$. Accordingly, in Assumption~\ref{as},~\ref{as1} states that $\di$ does not participate in the game during this period.~\ref{as2} ensures that all attackers are initially located outside the target region $\mathcal T$. Note that if $\vajb \geq \vdib$ and the capture radius of $\di$ satisfies $r_i=0$, then for certain target regions $\mathcal T$, such as those defined by $f(\mathbf x) = [0\ 1]\cdot\mathbf x$, $\aj$ can guarantee winning against $\di$. Therefore, many existing reach-avoid games impose the condition $\vdib>\vajb$. In the same spirit,~\ref{as3} requires that $\vdib>\vajb$ for $i=1,2,\dots,N$ and $j=1,2,\dots,M$. Since $\di$ does not participate in the game before $\tau_i$,~\ref{as4} guarantees that throughout the interval $[0, \tau_i]$, attacker $\aj$ will not enter the capture range of $\di$ regardless of its motion, thereby preventing $\aj$ from being captured before $\di$ joins the game. Moreover, due to the information delay $\tau_i$, the best $\di$ can do before capturing $\aj$ is to move to the position previously occupied by $\aj$, namely, to achieve $\mathbf x_{\di}^t=\mathbf x_{\aj}^{t-\tau_i}$ . If $\di$ still fails to capture $\aj$ at that moment, then it can only continue tracking the position of $\aj$ from $\tau_i$ time units earlier and will never be able to capture it, causing the game to continue indefinitely. Therefore,~\ref{as5} ensures that the game terminates in finite time, i.e., $t_f < \infty$.

\section{One-attacker-one-defender Reach-Avoid Game with Information Delay}\label{1vs1}
In this section, we will analyze a reach-avoid game with information delay involving one attacker and one defender, i.e., $N=M=1$. We present a method for determining whether the attacker can enter the target region before being captured by the defender, thereby characterizing the attacker’s winning condition. If the defender wins, then, in accordance with the objectives of both players, we derive the capture location, i.e., the coordinates of the capture point, together with the corresponding strategies of the attacker and the defender, and establish that these strategies form an equilibrium. Since $N = M = 1$, for notational simplicity, we omit the subscripts $i$ and $j$ in this section, and the cost function $J$ of this game is given by
\begin{equation}
    J=\min_{0<t\le t_f}\min_{\mathbf x\in\mathcal T}\|\mathbf x_A^t-\mathbf x\|.
    \label{eq:costfunction}
\end{equation}
Therefore, if $J=0$, the attacker $A$ wins; otherwise, $A$ aims to minimize $J$, whereas $D$ aims to maximize it.

\subsection{Analysis of Attack Region}\label{sub3a}
Based on existing studies on reach-avoid games~\cite{yan2023multiplayer}, a standard approach to analyzing such problems is to characterize the attack region $\mathcal R_A$, defined as the set of all points that $A$ can guarantee reaching before being captured by $D$, i.e.,
\begin{equation}
    \mathcal R_A=\{\mathbf x\in\mathbb R^2\,|\,\frac{\|\mathbf x_A-\mathbf x\|}{\bar v_A}\le\frac{\|\mathbf x_D-\mathbf x\|-r}{\bar v_D}\}.
    \label{eq:standardattackregion}
\end{equation}
If the attack region intersects with the target region, i.e., $\mathcal R_A\cap\mathcal T\ne\varnothing$, then $A$ wins; otherwise, $D$ wins. When $D$ can guarantee a win, the capture point $\mathbf x_f$, where $A$ will be captured by $D$, is given by the point in $\mathcal R_A$ closest to the target region, i.e.,
\begin{equation*}
    \mathbf x_f=\arg\min_{\mathbf x\in\mathcal R_A}\min_{\mathbf x_0\in\mathcal T}\|\mathbf x-\mathbf x_0\|.
\end{equation*}
In this way, the optimal strategies for the standard reach-avoid game can be obtained.
\begin{lemma}{\rm(Optimal Strategies for the Standard Reach-Avoid Game~\cite{yan2023multiplayer})}
    The optimal strategies of both players are to move toward the capture point at their respective maximum speeds, which constitute optimal strategies in the sense of a Nash equilibrium.\label{lm:priliminary}
    \label{lm:standardoptimalstrategies}
\end{lemma}

In this case, the optimal strategies of both players are to move toward the capture point at their respective maximum speeds. These strategies have been shown to constitute a Nash equilibrium. Motivated by this idea, we likewise seek to characterize the attack region $\mathcal R_A$ for the reach-avoid game with information delay in order to analyze the game.

Due to the presence of information delay, before characterizing $\mathcal R_A$, we must first determine which state variables govern the strategies of both players at time $t$. In standard reach-avoid games without information delay, the strategies of both players at time $t$ depend only on their current state variables $\mathbf x_D^t$ and $\mathbf x_A^t$. In the presence of information delay, however, the situation is different for both $D$ and $A$. For $D$, at time $t$, the only information available consists of its own current position and the position of $A$ at time $t-\tau$; earlier information has no effect on the current strategy, while later information is unavailable. Therefore, the control inputs $v_D(t)$ and $\theta_D(t)$ of $D$ at time $t$ depend only on $\mathbf x_D^t$ and $\mathbf x_A^{t-\tau}$. For $A$, at time $t$ (here we assume that $D$ has already entered the game; the case in which $D$ has not yet participated will be discussed in Section~\ref{sub3b}), in addition to observing the current positions of both players,  $A$ also knows that $D$ is subject to  an information delay. Hence, $A$ knows that the motion of $D$ over the next $\tau$ time units depends only on the positions of $A$ over the interval $[t-\tau, t]$ and the positions of $D$ over $[t, t+\tau]$, and is independent of $A$’s current strategy. Therefore, once the optimal strategies of both players are specified, $A$ can predict the position of $D$ at time $t+\tau$. Consequently, the control inputs $v_A(t)$ and $\theta_A(t)$ of $A$ at time $t$ depend only on $\mathbf x_A^t$ and $\mathbf x_D^{t+\tau}$.




\begin{figure}[!t]
\centerline{\includegraphics[width=0.6\columnwidth]{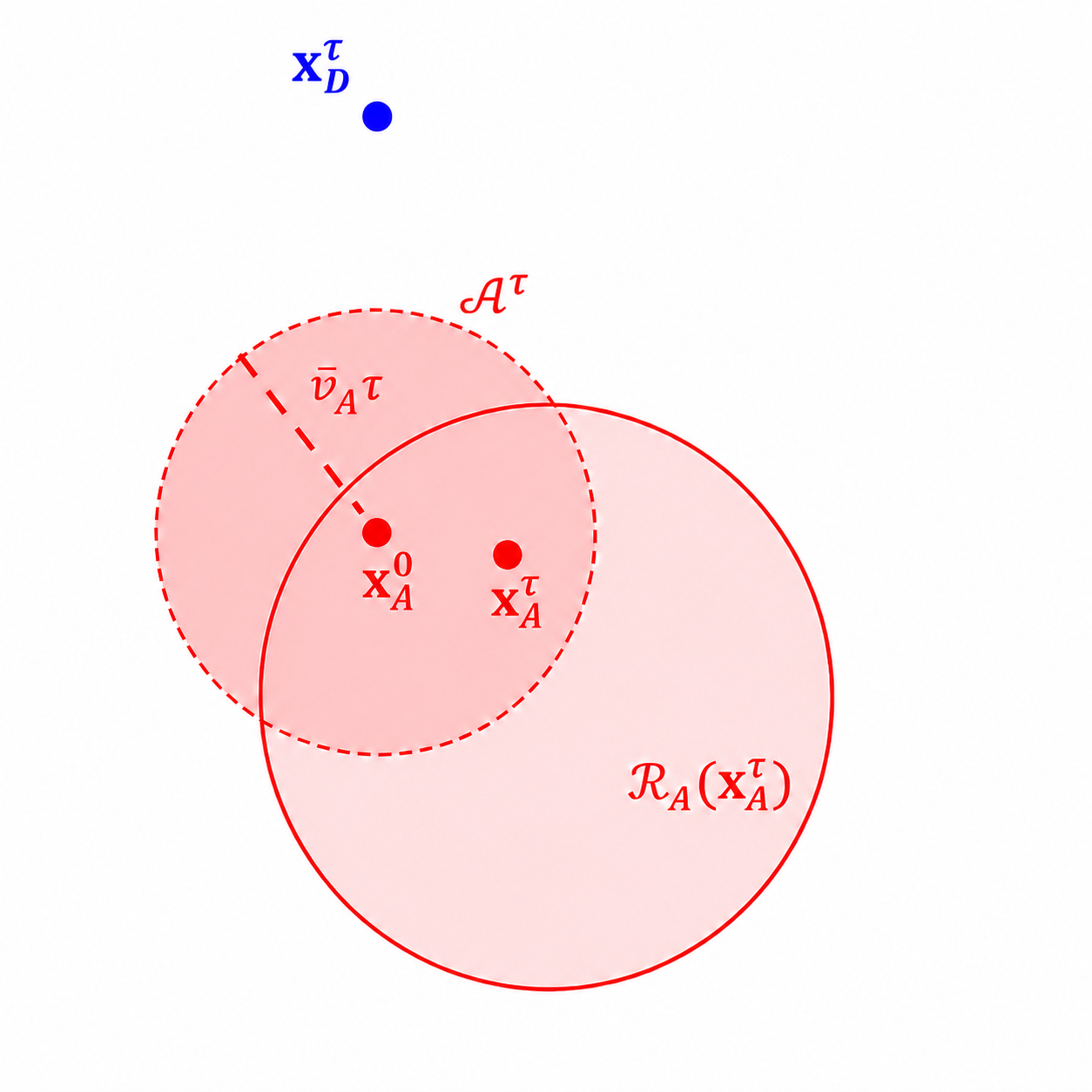}}
\caption{The set $\mathcal A^\tau$ of points that $A$ can reach at time $\tau$, and the set $\mathcal R_A(\mathbf x_A^\tau)$ of all points that $A$ can reach no later than $D$ when $A$ chooses $\mathbf x_A^\tau$ as its position at time $\tau$.}
\label{fig2_2a}
\end{figure}

\begin{figure}[!t]
\centerline{\includegraphics[width=0.6\columnwidth]{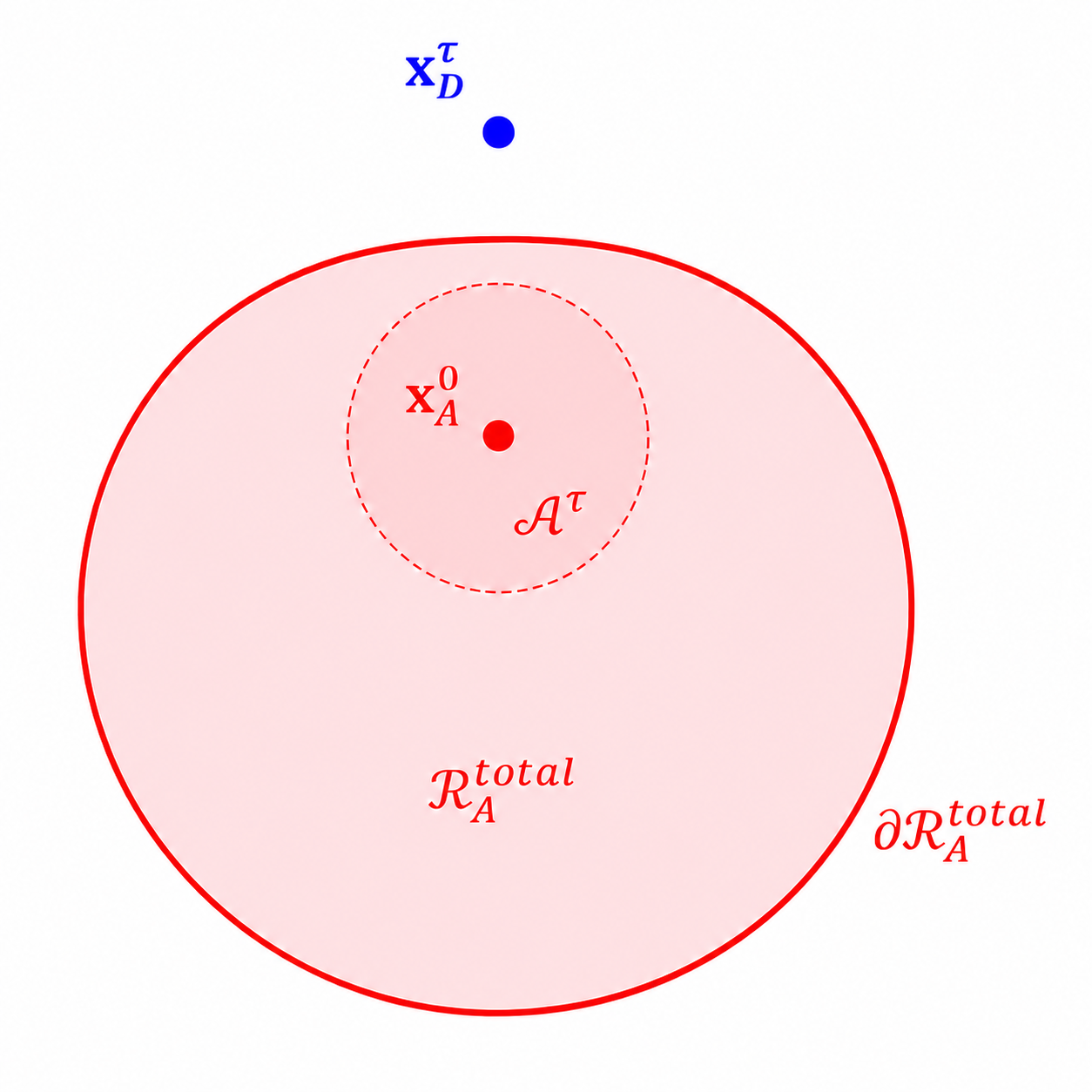}}
\caption{The attack region $\mathcal R_A^{total}$ and the boundary of attack region $\partial\mathcal R_A^{total}$.}
\label{fig2_2b}
\end{figure}

Next, we analyze the winning conditions and the optimal strategies for the reach-avoid game with information delay. Without loss of generality, we set the analysis time to be $t=0$, corresponding to the initial moment of the game, where $A$ enters the game at time $0$ and $D$ enters the game at time $\tau$. As in standard reach-avoid games, where the attack region must first be characterized, in the presence of information delay we also need to determine the set of all points that $A$ can guarantee reaching no later than $D$. To this end, we first consider the initial time interval $[0, \tau]$, during which $A$ can move freely over the entire two-dimensional plane. Let $\mathcal A^\tau$ denote the set of all points that $A$ can reach at time $\tau$. Since the maximum speed of $A$ is $\bar v_A$, it follows that $\mathcal A^\tau=\{\mathbf x\,|\,\|\mathbf x-\mathbf x_A^0\|\le\bar v_A\tau\}$, that is, $\mathcal A^\tau$ is a disk centered at $\mathbf x_A^0$ with radius $\bar v_A\tau$. Denote its boundary by $\partial\mathcal A^\tau$, i.e., $\partial\mathcal A^\tau=\{\mathbf x\,|\,\|\mathbf x-\mathbf x_A^0\|=\bar v_A\tau\}$. After time $\tau$, $A$ is located at some point $\mathbf x_A^\tau\in\mathcal A^\tau$. At that moment, $D$ enters the game, and $A$ and $D$ thereafter engage in a standard reach-avoid game.  By~\eqref{eq:standardattackregion}, when $A$ is located at $\mathbf x_A^\tau$, its reachable set is given by
\begin{equation}
    \mathcal R_A(\mathbf x_A^\tau)=\{\mathbf x\in\mathbb R^2\,|\,\frac{\|\mathbf x_A^\tau-\mathbf x\|}{\bar v_A}\le\frac{\|\mathbf x_D^\tau-\mathbf x\|-r}{\bar v_D}\},
    \label{eq:attackregionaftertau}
\end{equation}
where $\mathbf x_D^\tau=\mathbf x_D^0$, since $D$ does not participate in the game and remains stationary during the  interval $[0,\tau]$, as shown in Fig.~\ref{fig2_2a}. Because $\mathbf x_A^\tau$ may be any point in $\mathcal A^\tau$, and is determined solely by $A$’s strategy over the interval $[0, \tau]$, the union $\mathcal R_A^{total}$ of all such reachable sets $\mathcal R_A(\mathbf x_A^\tau)$,  as $\mathbf x_A^\tau$ ranges over $\mathcal A^\tau$, constitutes the attack region of the reach-avoid game with information delay. Therefore, the attack region is given by
\begin{equation*}
    \mathcal R_A^{total}=\bigcup_{\mathbf x_A^\tau\in\mathcal A^\tau}\mathcal R_A(\mathbf x_A^\tau).
\end{equation*}
Moreover, we denote the boundary of this attack region by $\partial\mathcal R_A^{total}$, as shown in Fig.~\ref{fig2_2b}. Next, we will provide the strategy that enables $A$ to reach $\partial\mathcal R_A^{total}$.

\begin{lemma}{\rm(Strategy for the Attacker to Reach the Boundary of the Attack Region)}
    The strategy for $A$ to reach $\partial\mathcal R_A^{total}$ is that $A$ moves at maximum speed along a fixed direction, i.e., $v_A=\bar v_A$ and $\theta_A$ is constant.\label{lm:necessaryconditionofboundary}
\end{lemma}
\begin{proof}
    For a point $\mathbf x$ on the boundary of $\mathcal R_A^{total}$, since $\mathcal R_A^{total}$ is the union of a collection of $\mathcal R_A(\mathbf x_A^\tau)$, $\mathbf x$ must lie on the boundary of some $\mathcal R_A(\mathbf x_A^\tau)$ and not in the interior of any other $\mathcal R_A(\mathbf x_A^\tau)$. From~\eqref{eq:attackregionaftertau}, we can express the boundary of $\mathcal R_A(\mathbf x_A^\tau)$ analytically as
    \begin{equation}
        h(\mathbf x,\mathbf x_A^\tau)=\|\mathbf x_D^\tau-\mathbf x\|-\frac{\bar v_D}{\bar v_A}\|\mathbf x_A^\tau-\mathbf x\|-r=0,
        \label{eq:boundaryofstandardattackregion}
    \end{equation}
    and points inside $\mathcal R_A(\mathbf x_A^\tau)$ satisfy $h(\mathbf x,\mathbf x_A^\tau)>0$. Therefore, a point $\mathbf x$ on $\partial\mathcal R_A^{total}$ must satisfy
    \begin{equation*}
        \max_{\mathbf x_A^\tau\in\mathcal A^\tau} h(\mathbf x,\mathbf x_A^\tau)=0,
    \end{equation*}
    that is, only for certain $\mathbf x_A^\tau$ does $\mathbf x$ satisfy $h(\mathbf x,\mathbf x_A^\tau)=0$, meaning $\mathbf x$ lies on the boundary of $\mathcal R_A(\mathbf x_A^\tau)$, while for other $\mathbf x_A^\tau$, $\mathbf x$ satisfies $h(\mathbf x,\mathbf x_A^\tau)<0$, meaning $\mathbf x$ lies outside $\mathcal R_A(\mathbf x_A^\tau)$. Therefore, for a point $\mathbf{x}$ on the boundary, $\mathbf x_A^\tau$ must be the $\mathbf x_A^\tau$ that maximizes $h(\mathbf x,\mathbf x_A^\tau)$. In~\eqref{eq:boundaryofstandardattackregion}, since $\|\mathbf x_D^\tau-\mathbf x\|$, $\bar v_D/\bar v_A$ and $r$ are fixed, maximizing $h(\mathbf x,\mathbf x_A^\tau)$ is equivalent to minimizing $\|\mathbf x_A^\tau-\mathbf x\|$, i.e., we need to find the $\mathbf x_A^\tau$ closest to $\mathbf x$. Since $\mathcal A^\tau$ is a disk centered at $\mathbf x_A^0$, the $\mathbf x_A^\tau$ closest to $\mathbf x$ must lie on the line segment connecting $\mathbf x$ and $\mathbf x_A^0$ and on the boundary of the disk $\mathcal A^\tau$. In order for $A$ to reach a point $\mathbf x_A^\tau$ on $\partial\mathcal A^\tau$ at time $\tau$, it must move at maximum speed toward $\mathbf x_A^\tau$. After time $\tau$, according to the optimal strategies in a standard reach-avoid game, both $A$ and $D$ must move at maximum speed toward $\mathbf x$. Combining these two stages, to reach a point $\mathbf x$ on $\partial\mathcal R_A^{total}$, $A$’s strategy over the entire course of the game should be to move at maximum speed directly toward $\mathbf x$, i.e., $v_A=\bar v_A$ and $\theta_A$ is constant.
\end{proof}

From Lemma~\ref{lm:necessaryconditionofboundary}, we know that $A$’s strategy is to move in a straight line at maximum speed toward the capture point. Once $D$ enters the game, it follows the optimal strategy of the standard reach-avoid game; according to Lemma~\ref{lm:priliminary}, i.e., moving in a straight line at maximum speed toward the same capture point. Therefore, throughout the game, both $A$ and $D$ move at maximum speed in a straight line toward the same capture point, with $A$ being $\tau$ time units ahead of $D$. Based on this observation, we can now derive an analytical characterization of the set $\mathcal R_A^{total}$.

\begin{lemma}{\rm(Analytical Representation of the Attack Region)}
    In the reach-avoid game with information delay, the attack region admits the following analytical representation:
    \begin{equation}
        \begin{aligned}
            \mathcal R_A^{total}=\{\mathbf x\in\mathbb R^2\,|\,g(\mathbf x,\mathbf x_A^0)=\|\mathbf x_D^\tau-\mathbf x\|-\frac{\bar v_D}{\bar v_A}\|\mathbf x_A^0-\mathbf x\|\\+\bar v_D\tau-r\ge0\}.
        \end{aligned}
        \label{eq:AttackRegion}
    \end{equation}
    \label{lm:AnalyticalRepresentationoftheAttackRegion}
\end{lemma}
\begin{proof}
    From Lemmas~\ref{lm:priliminary} and~\ref{lm:necessaryconditionofboundary}, a point $\mathbf x$ on $\partial\mathcal R_A^{total}$ satisfies the following property: when $A$ moves $\tau$ time units ahead of $D$ and both players then move toward $\mathbf x$ along straight-line paths at maximum speed, $A$ will be captured by $D$ at $\mathbf x$. That is
    \begin{equation}
        \frac{\|\mathbf x_A^0-\mathbf x\|}{\bar v_A}-\tau=\frac{\|\mathbf x_D^\tau-\mathbf x\|-r}{\bar v_D},
        \label{eq:geometricproperty}
    \end{equation}
    where $\mathbf x_D^\tau=\mathbf x_D^0$, since $D$ remains stationary during the time interval $[0,\tau]$. It then follows from~\eqref{eq:geometricproperty} that a point $\mathbf x$ on $\partial\mathcal R_A^{total}$ satisfies
    \begin{equation*}
        g(\mathbf x,\mathbf x_A^0)=\|\mathbf x_D^\tau-\mathbf x\|-\frac{\bar v_D}{\bar v_A}\|\mathbf x_A^0-\mathbf x\|+\bar v_D\tau-r=0.
    \end{equation*}
    Moreover, for any point in $\mathcal R_A^{total}$, attacker $A$ can guarantee reaching it no later than $D$. Therefore, for any $\mathbf x\in\mathcal R_A^{total}$, we have $g(\mathbf x,\mathbf x_A^0)\ge0$.
\end{proof}

Lemma~\ref{lm:AnalyticalRepresentationoftheAttackRegion} provides an analytical characterization of the attack region. We  now proceed to study its convexity. As a preliminary step, we first study a class of closed and bounded planar regions and characterize their convexity.

\begin{lemma}{\rm(Convexity Criterion)}
    Suppose that $\alpha>1$ and $-\|\mathbf x_1-\mathbf x_2\|<b<(\alpha-1)\|\mathbf x_1-\mathbf x_2\|$. Then the set
    \begin{equation}
        S=\{\mathbf x\in\mathbb R^2\,|\,\|\mathbf x-\mathbf x_1\|-\alpha\|\mathbf x-\mathbf x_2\|+b\ge0\}
        \label{eq:convexset}
    \end{equation}
    is convex.
    \label{lm:ConvexityCriterion}
\end{lemma}
\begin{proof}
    For computational convenience, we place $\mathbf x_2$ at the origin and choose the direction from $\mathbf x_2$ to $\mathbf x_1$ as the positive $x$-axis, i.e., we set $\mathbf x_1=(c,0)$, $\mathbf x_2=(0,0)$, where $c=\|\mathbf x_1-\mathbf x_2\|$. Define $k(\mathbf x)=\|\mathbf x-\mathbf x_1\|-\alpha\|\mathbf x-\mathbf x_2\|+b$. Then any boundary point $\mathbf x$ satisfies $k(\mathbf x)=0$. Let $\mathbf x$ be a boundary point with polar coordinates $(\rho\cos\theta,\rho\sin\theta)$. Then $\rho$ and $\theta$ satisfy
    \begin{equation*}
        (\alpha^2-1)\rho^2-2(\alpha b-c\cos\theta)\rho+b^2-c^2=0.
    \end{equation*}
    Accordingly, we obtain an analytical expression for $\rho$ as a function of $\theta$:
    \begin{equation}
        \begin{aligned}
            &\rho(\theta)\\=&\frac{(\alpha b-c\cos\theta)+\sqrt{(\alpha b-c\cos\theta)^2-(\alpha^2-1)(b^2-c^2)}}{\alpha^2-1}.
        \end{aligned}
        \label{eq:relationshipbetweenrandtheta}
    \end{equation}
    Since all terms involving $\theta$ in~\eqref{eq:relationshipbetweenrandtheta} appear through $\cos\theta$, we introduce the substitution $w=\cos\theta$. The resulting analytical expression for $\rho$ as a function of $w$ is given by
    \begin{equation}
        \rho(w)=\frac{(\alpha b-cw)+\sqrt{(\alpha b-cw)^2-(\alpha^2-1)(b^2-c^2)}}{\alpha^2-1}.
        \label{eq:relationshipbetweenrandw}
    \end{equation}
    With the relationship between $\rho$ and $\theta$ established, we can compute the curvature at a boundary point $\mathbf x$ as 
    \begin{equation}
        \kappa=\frac{\rho^2+2\rho_\theta^2-\rho\rho_{\theta\theta}}{(\rho^2+\rho_\theta^2)^{\frac{3}{2}}},
        \label{eq:Curvatureoftheta}
    \end{equation}
    where $\rho_\theta$ and $\rho_{\theta\theta}$ denote the first- and second-order derivatives of $\rho$ with respect to $\theta$, respectively.

    From~\eqref{eq:relationshipbetweenrandw}, we can compute the first- and second-order derivatives of $\rho$ with respect to $w$ as
    \begin{equation}
            \rho_w=-\frac{c\rho}{Q}, \rho_{ww}=\frac{c(\rho_w(\alpha b-cw)+c\rho)}{Q^2},
        \label{eq:derivativeofrwithrespecttow}
    \end{equation}
    where $Q=(\alpha^2-1)\rho-(\alpha b-cw)=\sqrt{(\alpha b-cw)^2-(\alpha^2-1)(b^2-c^2)}>0$. Moreover,
    \begin{equation}
        \begin{aligned}
            \rho_\theta&=-\rho_w\sin\theta,\\\rho_{\theta\theta}&=-\rho_w\cos\theta+\rho_{ww}(1-\cos^2\theta).
        \end{aligned}
        \label{eq:derivativeofrwithrespecttotheta}
    \end{equation}
    We then substitute~\eqref{eq:derivativeofrwithrespecttotheta} into~\eqref{eq:Curvatureoftheta} and obtain
    \begin{equation}
        \kappa(w)=\frac{\rho^2+2(1-w^2)\rho_w^2+w\rho\rho_w-(1-w^2)\rho\rho_{ww}}{(\rho^2+(1-w^2)\rho_w^2)^{\frac{3}{2}}}.
        \label{eq:Curvatureofw}
    \end{equation}
    Substituting~\eqref{eq:derivativeofrwithrespecttow} into~\eqref{eq:Curvatureofw} and differentiating with respect to $w$ yield
    \begin{equation*}
        \kappa'(w)=-K(w)\Psi(w),
    \end{equation*}
    where
    \begin{equation*}
        \begin{aligned}
            K(w)&=\frac{2(1-w^2)c^2\rho^3(\alpha^2-1)}{(\rho^2+(1-w^2)\rho_w^2)^{\frac{5}{2}}Q^3},\\\Psi(w)&=Q+\frac{(\alpha b-cw)((\alpha^2-1)\rho^2+(1-w^2)c^2)}{Q^2}.
        \end{aligned}
    \end{equation*}
    For $w\in[-1,1]$, $\alpha>1$ and $-c<b<(\alpha-1)c$, we have $K(w)>0$ and $\Psi(w)>0$. Thus, $\kappa'(w)<0$ for $w \in [-1, 1]$, indicating that the curvature of the boundary of $S$ is strictly decreasing with respect to $w$. Therefore, for $w\in[-1,1]$, we have
    \begin{equation*}
        \kappa(w)\ge\kappa(1)=\frac{(\alpha+1)((\alpha-1)c-b)}{(b+c)(\alpha c-b)}>0,
    \end{equation*}
    which means the curvature of all points on the boundary of $S$ is strictly positive, i.e., $S$ is a convex set.
\end{proof}

Based on Lemma~\ref{lm:AnalyticalRepresentationoftheAttackRegion} and Lemma~\ref{lm:ConvexityCriterion}, we can determine the convexity of the attack region~\eqref{eq:AttackRegion} in the reach-avoid game with information delay.

\begin{theorem}{\rm(Convexity of the Attack Region in the Reach-Avoid Game with Information Delay)}
Under Assumption~\ref{as}, the attack region~\eqref{eq:AttackRegion} in the reach-avoid game with information delay is convex.
\end{theorem}
\begin{proof}
    By comparing the attack region~\eqref{eq:AttackRegion} and set~\eqref{eq:convexset}, we observe that $\alpha=\bar v_D/\bar v_A>1$ and $b=\bar v_D\tau-r$. According to Lemma~\ref{lm:ConvexityCriterion}, to prove the convexity of~\eqref{eq:AttackRegion}, it suffices to show that
    \begin{equation}
        -\|\mathbf x_A-\mathbf x_D\|<\bar v_D\tau-r<(\bar v_D/\bar v_A-1)\|\mathbf x_A-\mathbf x_D\|.
        \label{eq:convexofAR}
    \end{equation}
    
    According to Assumption~\ref{as}, we have $r>\bar v_A\tau$ and $\|\mathbf x_A-\mathbf x_D\|>r+\bar v_A\tau$. Therefore, we have
    \begin{equation*}
        \begin{aligned}
            \bar v_D\tau-r&<(\frac{\bar v_D}{\bar v_A}-1)(r+\bar v_A\tau)<(\frac{\bar v_D}{\bar v_A}-1)\|\mathbf x_A-\mathbf x_D\|,\\\bar v_D\tau-r&>-r>-\|\mathbf x_A-\mathbf x_D\|+\bar v_A\tau>-\|\mathbf x_A-\mathbf x_D\|,
        \end{aligned}
    \end{equation*}
    which means~\eqref{eq:convexofAR} is proven. Therefore,~\eqref{eq:AttackRegion} is a convex set.
\end{proof}

    With the analytical characterization of the attack region and the proof of its convexity in place, we are now ready to determine the winning conditions of the game. If $\mathcal R_A^{total}\cap\mathcal T\ne\emptyset$, then $A$ can guarantee a win, i.e., regardless of the strategy adopted by $D$, there always exists a strategy for $A$ that ensures that it reaches $\mathcal T$ before being captured. On the other hand, if $\mathcal R_A^{total}\cap\mathcal T=\emptyset$, then $D$ can guarantee a win, i.e., regardless of the strategy adopted by $A$, there always exists a strategy for $D$ that ensures $A$ is captured before reaching $\mathcal T$. Moreover, algorithms for determining whether two analytically represented convex sets intersect are already well established.

\subsection{Optimal Strategies}\label{sub3b}
We now have an analytical characterization of the attack region. The area covered by this region is precisely the dominant region of $A$, since at every point within it, $A$ can guarantee arriving earlier than $D$. We also describe how to determine the winning conditions of the game. Therefore, when $A$ can guarantee a win, it suffices for $A$ to move toward the region $\mathcal R_A^{total}\cap\mathcal T$, and $D$ cannot capture $A$ regardless of its strategy. When $D$ can guarantee a win, by contrast, further analysis is required to identify equilibrium strategies under which the minimum distance between $A$ and the target region is optimized during the game. Hence, the subsequent discussion focuses on the case when $D$ can guarantee a win.

According to existing results on reach-avoid games, since the objective of $A$ is to minimize its minimum distance to the target region during the game, it will seek the point in its dominant region that is closest to the target region, i.e., the point on $\mathcal R_A^{total}$ closest to $\mathcal T$. Therefore, the capture point $\mathbf x^*$ in this problem, that is, the position at which $A$ is captured by $D$, is the point on $\mathcal R_A^{total}$ closest to $\mathcal T$, i.e., the closest point to the target region that $A$ can guarantee reaching no later than $D$. Specifically, let $(\mathbf x^*,\mathbf y^*)$ be the optimal solution of the following optimization problem, denoted by $\mathcal P^1(\mathbf x_A^0,\mathbf x_D^\tau,\tau)$:
\begin{equation}
    \begin{aligned}
        \min_{\mathbf x\in\mathbb R^2,\mathbf y\in\mathbb R^2}\quad&\|\mathbf x-\mathbf y\|\\s.t.\quad&\|\mathbf x_D^\tau-\mathbf x\|-\frac{\bar v_D}{\bar v_A}\|\mathbf x_A^0-\mathbf x\|+\bar v_D\tau-r\ge0,\\&f(\mathbf y)\le0,
    \end{aligned}
    \label{eq:optimizationproblem}
\end{equation}
which is a convex optimization problem. Then  $\mathbf x^*$ is the capture point. In other words, $\mathbf x^*$ is the point in $\mathcal R_A^{\rm total}$ that is closest to the target region $\mathcal T$; equivalently, it is the closest point to $\mathcal T$ that $A$ can guarantee reaching no later than $D$.

Based on the preceding analysis of both players’ strategies and~\eqref{eq:optimizationproblem}, we can characterize the strategies of the two players at time $t\ge\tau$ as 
\begin{equation}
    \begin{aligned}
        v_D^*&=\bar v_D,\quad\theta_D^*=\arctan\frac{y_D^*-y_D^t}{x_D^*-x_D^t},\\v_A^*&=\bar v_A,\quad\theta_A^*=\arctan\frac{y_A^*-y_A^t}{x_A^*-x_A^t},
    \end{aligned}
    \label{eq:optimalstrategyaftertau}
\end{equation}
where $\mathbf x_D^*=(x_D^*,y_D^*)$ and $\mathbf x_A^*=(x_A^*,y_A^*)$ are the first component of the optimal solutions to optimization problems $\mathcal P^1(\mathbf x_A^{t-\tau},\mathbf x_D^t,\tau)$ and $\mathcal P^1(\mathbf x_A^t,\mathbf x_D^{t+\tau},\tau)$, respectively.

We note that $\mathbf x_D^{t+\tau}$, which represents the position of $D$ after a delay of $\tau$, is unknown at time $t$. Therefore, in order to determine its current strategy, $A$ must use its own historical information over the interval $[t-\tau,t]$, together with the currently observed opponent position $\mathbf x_D^t$, to infer $D$’s motion over the next $\tau$ time units. Based on this inference, $A$ can estimate $\mathbf x_D^{t+\tau}$ and thus determine its current strategy.

In addition, strategy~\eqref{eq:optimalstrategyaftertau} applies only when $t\ge\tau$. For $t<\tau$, $D$ has not yet entered the game and thus has no strategy, whereas the strategy of $A$ is given by
\begin{equation}
    v_A^*=\bar v_A,\quad\theta_A^*=\arctan\frac{y_A^*-y_A^t}{x_A^*-x_A^t},
    \label{eq:optimalstrategybeforetau}
\end{equation}
where $\mathbf x_A^*=(x_A^*,y_A^*)$ is the first block component of the optimal solution of the optimization problem $\mathcal P^1(\mathbf x_A^t,\mathbf x_D^0,\tau-t)$.

\subsection{Equilibrium Verification}\label{sub3c}
Now that we have obtained the strategies~\eqref{eq:optimalstrategyaftertau} and~\eqref{eq:optimalstrategybeforetau} for the two players at time $t$, we next examine whether they constitute an equilibrium of the differential game. Unlike standard reach-avoid games, the delayed-information setting features sequential entry of the two players, which aligns more naturally with the structure of dynamic games in game theory. Accordingly, to prove that strategies~\eqref{eq:optimalstrategyaftertau} and~\eqref{eq:optimalstrategybeforetau} constitute an equilibrium, it is necessary to establish their optimality in the sense of a subgame-perfect Nash equilibrium, i.e., they must be optimal at every stage of the game (that is, in every subgame).

\begin{theorem}{\rm(Subgame-Perfect Equilibrium)}
    In the reach-avoid game with information delay, strategies~\eqref{eq:optimalstrategyaftertau} and~\eqref{eq:optimalstrategybeforetau} are optimal strategies in the sense of a subgame-perfect Nash equilibrium.
\end{theorem}
\begin{proof}
    For the reach-avoid game with information delay, the game can be naturally divided into two stages: $t\ge\tau$ and $t<\tau$, corresponding to two subproblems. We then show,  via backward induction, that the strategies for both stages are optimal in the sense of a Nash equilibrium.

    From the analysis in Subsection~\ref{sub3a}, when $t\ge\tau$, the game reduces to a standard reach-avoid game. Accordingly, the optimal strategies in this stage are exactly those characterized in Lemma~\ref{lm:standardoptimalstrategies}, and they are optimal in the sense of a Nash equilibrium.

    When $t<\tau$, $A$ can move freely in $\mathbb R^2$. Hence, its optimal strategy is to choose its position at time $\tau$ such that the capture point $\mathbf x^*$ induced in the stage $t\ge\tau$ is as close to the target region as possible. Let $\mathcal A^\tau$ be the set of all positions reachable by $A$ by time $\tau$. Then, for $t<\tau$, $A$ moves toward the capture point $\mathbf x^*$, where $\mathbf x^*$ is determined as the first component of $(\mathbf x^*,\mathbf y^*)$, which is the optimal solution to 
    \begin{equation}
        \begin{aligned}
            \min_{\mathbf x\in\mathbb R^2,\mathbf y\in\mathbb R^2}\quad&\|\mathbf x-\mathbf y\|\\s.t.\quad&\mathbf x_A^\tau\in\mathcal A^\tau,\\\quad&\mathbf x\in\mathcal R_A(\mathbf x_A^\tau),\\&f(\mathbf y)\le0.
        \end{aligned}
        \label{eq:optimizationproblemofSPNE}
    \end{equation}

    Since $\mathcal R_A^{total}$ is the union of $\mathcal R_A(\mathbf x_A^\tau)$ over all $\mathbf x_A^\tau\in\mathcal A^\tau$,  problems~\eqref{eq:optimizationproblemofSPNE} and~\eqref{eq:optimizationproblem} are  equivalent, in the sense that they admit the same optimal solution and yield the same strategies. Therefore,  because problem~\eqref{eq:optimizationproblemofSPNE} defines a strategy that is optimal in the sense of a subgame-perfect Nash equilibrium, strategies~\eqref{eq:optimalstrategyaftertau} and~\eqref{eq:optimalstrategybeforetau} are also optimal in this sense.
\end{proof}

Therefore, we obtain subgame-perfect Nash equilibrium strategies for the one-attacker-one-defender reach-avoid game with information delay.

\section{Multiplayer Reach-Avoid Game and Assignment with Information Delay}\label{multi}
In this section, building upon the analytical framework and results established for the one-attacker-one-defender reach-avoid game in Section~\ref{1vs1}, we extend the problem first to the one-attacker-multiple-defender case, and then to the more general setting involving $M$ attackers and $N$ defenders.

\subsection{One-Attacker-Multiple-Defender Reach-Avoid Games with Information Delay}
We begin with the reach-avoid game with one attacker and multiple defenders, i.e., $M = 1$. As before, for notational simplicity, we omit the subscript $j$ in this subsection. The players’ objectives and the cost function are exactly the same as in the one-attacker-one-defender setting. Specifically, $A$ aims to reach the target region $\mathcal T$, or if captured beforehand, to be as close to $\mathcal T$ as possible. Each defender $\di$ aims to prevent $A$ from entering $\mathcal T$ and, throughout the game, to keep $A$ as far from $\mathcal T$ as possible. The cost function is again given by~\eqref{eq:costfunction}. Thus, if $J=0$, then $A$ wins; otherwise, $A$ aims to minimize $J$, whereas $\di$ aims to maximize it. We assume throughout that the indices $i$ are ordered according to $\tau_i$ in ascending order.

\begin{figure}[!t]
\centerline{\includegraphics[width=0.8\columnwidth]{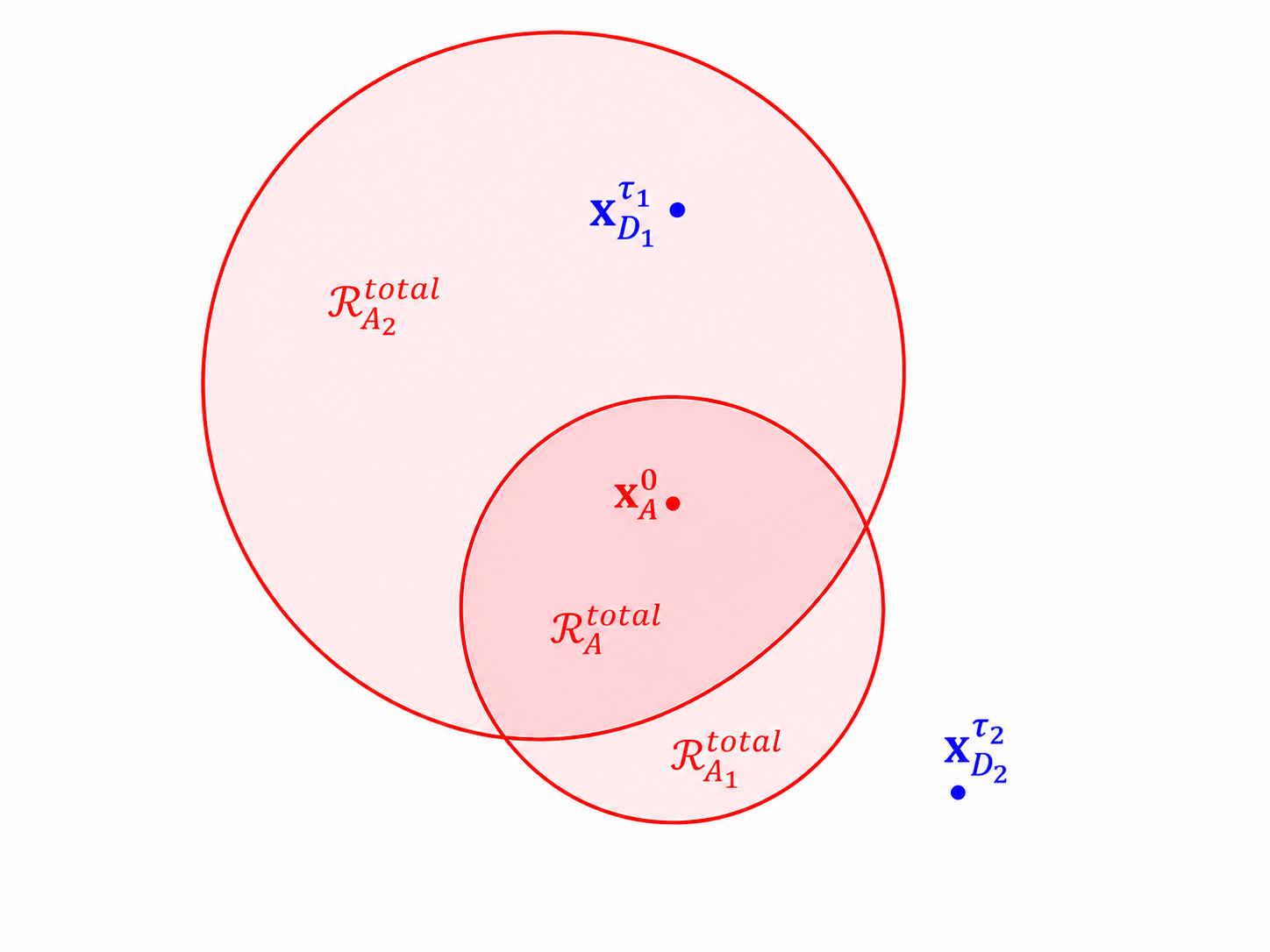}}
\caption{The set $\mathcal R_
{A_1}^{total}$ of points that $A$ can reach no later than $D_1$, the set $\mathcal R_
{A_2}^{total}$ of points that $A$ can reach no later than $D_2$, and the set $\mathcal R_
A^{total}$ of points that $A$ can reach no later than both defenders simultaneously.}
\label{fig2_3}
\end{figure}

We next characterize the attack region of this reach-avoid game, i.e., the set of points that $A$ can reach no later than any defender $\di$. By definition, a point $\mathbf x$ lies in the attack region if $A$ reaches $\mathbf x$ no later than all defenders $\di$, i.e., $\mathbf x\in\mathcal R_{A_i}^{total}$ for each $i =1,2,\dots,N$, where $\mathcal R_{A_i}^{total}$ denotes the attack region of the one-attacker-one-defender reach-avoid game with information delay between $A$ and $\di$, i.e.,
\begin{equation*}
    \begin{aligned}
        \mathcal R_{A_i}^{total}=\{\mathbf x\in\mathbb R^2\,|\,g_i(\mathbf x,\mathbf x_A^0)=\|\mathbf x_{D_i}^{\tau_i}-\mathbf x\|-\frac{\bar v_{D_i}}{\bar v_A}\|\mathbf x_A^0-\mathbf x\|\\+\bar v_{D_i}\tau_i-r_i\ge0\}.
    \end{aligned}
\end{equation*}
Hence, $\mathbf x\in\bigcap_{i=1}^N\mathcal R_{A_i}^{total}$, and the attack region of the one-attacker-multiple-defender reach-avoid game with information delay can be characterized as
\begin{equation*}
    \mathcal R_A^{total}=\bigcap_{i=1}^N\mathcal R_{A_i}^{total},
\end{equation*}
as shown in Fig.~\ref{fig2_3}.

According to~\ref{as4} in Assumption~\ref{as}, we have
\begin{equation*}
    \begin{aligned}
        g_i(\mathbf x_A^0,\mathbf x_A^0)=&\|\mathbf x_{D_i}^{\tau_i}-\mathbf x_A^0\|+\bar v_{D_i}\tau_i-r_i\\=&\|\mathbf x_{D_i}^0-\mathbf x_A^0\|+\bar v_{D_i}\tau_i-r_i\\>&(\bar v_{D_i}+\bar v_A)\tau_i\\>&0.
    \end{aligned}
\end{equation*}
Hence, $\mathbf x_A^0\in\mathcal R_{A_i}^{total}$ for every $i =1,2,\dots,N$, which implies that $\mathbf x_A^0\in\bigcap_{i=1}^N\mathcal R_{A_i}^{total}$. Therefore, the attack region $\mathcal R_A^{total}$ is nonempty.

With the attack region $\mathcal R_A^{total}$ characterized, we can, similarly to the analysis in Subsection~\ref{sub3b},  conclude that the capture point in the one-attacker-multiple-defender reach-avoid game with information delay is given by the first component of the optimal solution of the optimization problem $\mathcal P^N(\mathbf x_A^0,\mathbf x_{D_1}^{\tau_1},\mathbf x_{D_2}^{\tau_2},\dots,\mathbf x_{D_N}^{\tau_N},\tau_1,\tau_2\dots,\tau_N)$:
\begin{equation}
    \begin{aligned}
        \min_{\mathbf x\in\mathbb R^2,\mathbf y\in\mathbb R^2}\quad&\|\mathbf x-\mathbf y\|\\s.t.\quad&g_i(\mathbf x,\mathbf x_A^0)\ge0,\quad \mathrm{for}\ \ i=1,2,\dots,N,\\&f(\mathbf y)\le0,
    \end{aligned}
    \label{eq:optimazationproblemformutiple}
\end{equation}
which is a convex optimization problem. Accordingly, the optimal strategy of the attacker $A$ at time $t$ is 
\begin{equation}
    v_A^*=\bar v_A,\quad\theta_A^*=\arctan\frac{y_A^*-y_A^t}{x_A^*-x_A^t},
    \label{eq:optimalstrategyfora}
\end{equation}
where $\mathbf x_A^*=(x_A^*,y_A^*)$ is the first component of the optimal solution of optimization problem $\mathcal P^N(\mathbf x_A^t,\mathbf x_{D_1}^{p_1},\mathbf x_{D_2}^{p_2},\dots,\mathbf x_{D_N}^{p_N},s_1,s_2\dots,s_N)$, where
\begin{equation*}
    \begin{aligned}
        p_k=\begin{cases}
            0,&t<\tau_k\\t+\tau_k,&t\ge\tau_k
        \end{cases}, s_k=\begin{cases}
            \tau_k-t,&t<\tau_k\\\tau_k,&t\ge\tau_k
        \end{cases},\\k=1,\dots,N.
    \end{aligned}
\end{equation*}
Similarly, the optimal strategy of defender $\di$ at time $t$ is 
\begin{equation}
    v_{D_i}^*=\bar v_{D_i},\quad\theta_{D_i}^*=\arctan\frac{y_{D_i}^*-y_{D_i}^t}{x_{D_i}^*-x_{D_i}^t},
    \label{eq:optimalstrategyford}
\end{equation}
where $\mathbf x_{D_i}^*=(x_{D_i}^*,y_{D_i}^*)$ is the first component of the optimal solution of the optimization problem $\mathcal P^N(\mathbf x_A^{t-\tau_i},\mathbf x_{D_1}^{q_1},\mathbf x_{D_2}^{q_2},\dots,\mathbf x_{D_N}^{q_N},s_1,s_2\dots,s_N)$, with
\begin{equation*}
    q_k=\begin{cases}
        0,&t<\tau_k\\t-\tau_i+\tau_k,&t\ge\tau_k
    \end{cases},\quad k=1,\dots,N.
\end{equation*}
Here, we note that for $t\ge\tau_k$, the strategy of $\di$ needs to predict the position of $D_k$ at time $t-\tau_i+\tau_k$. Similar to the strategy~\eqref{eq:optimalstrategyaftertau} of $A$ in Subsection~\ref{sub3b}, $D_i$ needs to use historical information to compute this prediction step by step starting from $t=t_c$, where $t_c=\max(0,t-\tau_1-\tau_2-\dots-\tau_N)$.

In this way, we obtain the optimal strategies~\eqref{eq:optimalstrategyfora} and~\eqref{eq:optimalstrategyford} for the one-attacker-multiple-defender reach-avoid game with information delay.

    We note that the function $g_i(\mathbf x,\mathbf x_A^0)$ defined in this paper satisfies the definition of a Pursuit Enclosure Function in~\cite[Definition 1]{yan2024multiplayer}. Consequently, our optimization problem~\eqref{eq:optimazationproblemformutiple} belongs to the class of problems  introduced in~\cite[Definition 4]{yan2024multiplayer}. It then follows from~\cite[Theorem 1]{yan2024multiplayer} that the same optimal solution can be obtained using at most two functions $g_i(\mathbf x,\mathbf x_A^0)$. In other words, in the one-attacker-multiple-defender reach-avoid game with information delay, at most two defenders are sufficient to determine the capture point under optimal play, and hence the corresponding optimal strategies of both players.

\subsection{N-vs-M Matching and Assignment}
Next, we consider the reach-avoid game with information delay involving $N$ defenders and $M$ attackers. In this setting, the optimization objective changes: rather than focusing on the minimum distance between an attacker $A_j$ and the target region during the game, the goal is to capture as many attackers as possible before the game ends. We assume that once a defender $\di$ captures an attacker $\aj$, it no longer participates in subsequent captures. Under this assumption, the problem can be formulated as a maximum matching problem. 

In~\cite{yan2019task}, Yan et al. studied the assignment problem for an $N$-versus-$M$ reach-avoid game in which the target region is a half-plane and the defenders have no capture radius. In particular, \cite[Theorem 5]{yan2019task} presents an optimization formulation corresponding to the maximum matching in that setting. In our case, the maximum matching formulation can be expressed in essentially the same form as that in~\cite[Theorem 5]{yan2019task}. The only difference is in the definition of the vector $\mathbf w$ (which corresponds to the vector $\mathbf r$ in~\cite[Theorem 5]{yan2019task}, and is renamed here to avoid confusion with the capture radius used in this paper). Specifically, the entries of $\mathbf w$ are defined as follows. Set $w_i^1(j)=0$ if $\di$ cannot guarantee capture of $\aj$ before it reaches $\mathcal T$; otherwise, set $w_i^1(j)=1$. Similarly, set $w_{i1,i2}^2(j)=0$ if $D_{i1}$ and $D_{i2}$ cannot guarantee capture of $\aj$ before it reaches $\mathcal T$; otherwise, set $w_{i1,i2}^2(j)=1$. We then obtain the prior-information vector $\mathbf w=[\mathbf w_1^1,\dots,\mathbf w_N^1,\mathbf w^2_{1,2},\dots,\mathbf w^2_{1,N},\mathbf w^2_{2,3},\dots,\mathbf w^2_{N-1,N}]$, where $\mathbf w_i^1=[\mathbf w_i^1(1),\dots,\mathbf w_i^1(M)]$ and $\mathbf w_{i1,i2}^2=[\mathbf w_{i1,i2}^2(1),\dots,\mathbf w_{i1,i2}^2(M)]$. In this way, we provide a solution method for the maximum matching problem in the reach-avoid game with information delay involving $N$ defenders and $M$ attackers, which is essentially consistent with the formulation in~\cite[Theorem 5]{yan2019task}.

\section{Numerical Simulations}\label{simu}
In this section, we present simulation results for both the one-on-one and the multi-agent settings to demonstrate the effectiveness of the proposed strategies. All simulations are produced using MATLAB R2023b. The hardware configuration is as follows: CPU: 13th Gen Intel® Core™ i9-13980HX @ 2.20 GHz, Memory: 16.0 GB RAM.

\subsection{One-vs-One Case}
\begin{figure*}[!t]
\centering
\subfigure[$A$ and $D$ both use the optimal strategies.]{\label{fig1a}\includegraphics[width=2.3in]{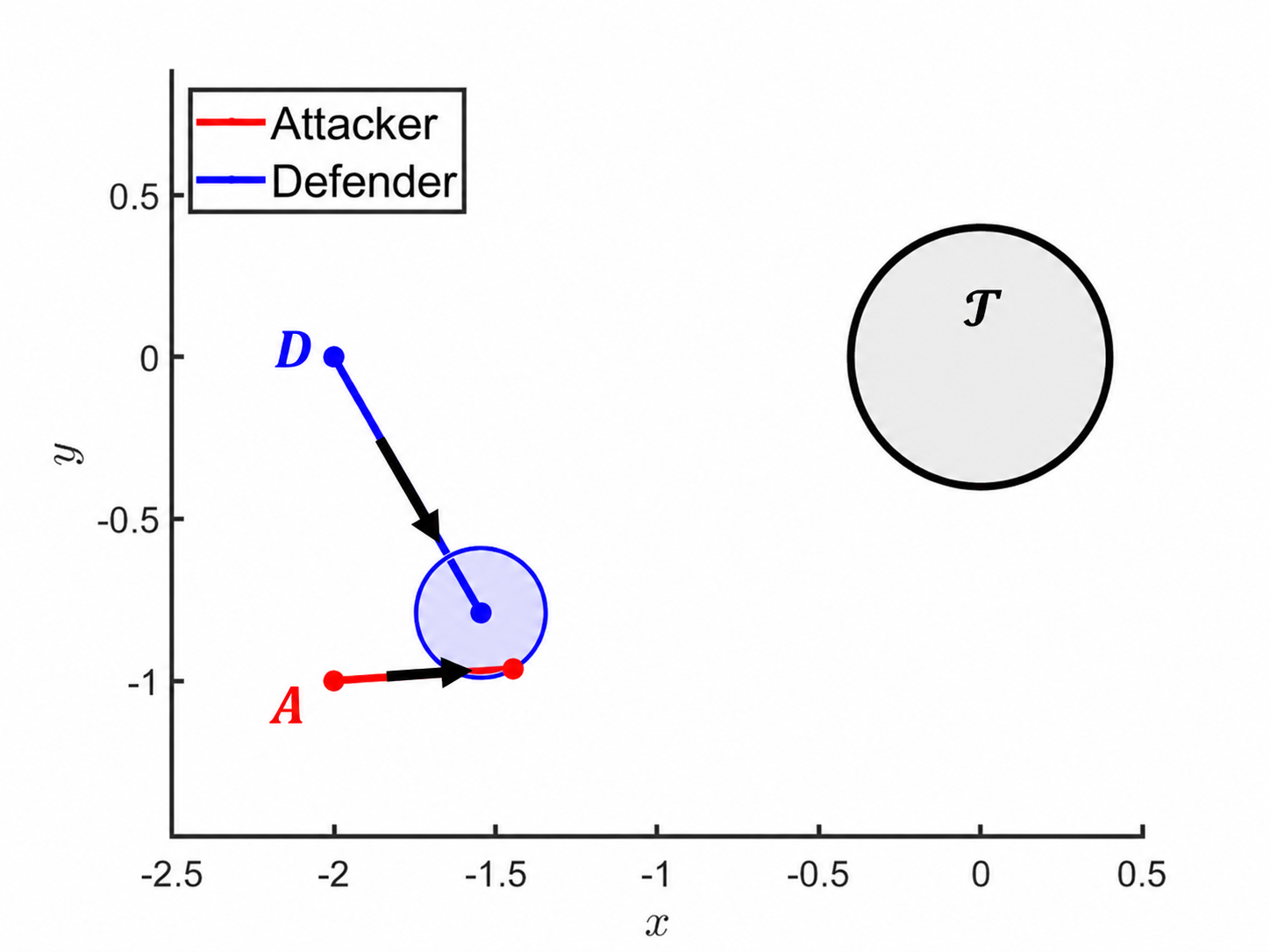}}
\subfigure[$A$ uses the pure-attack strategy while $D$ uses the optimal strategy.]{\label{fig1b}\includegraphics[width=2.3in]{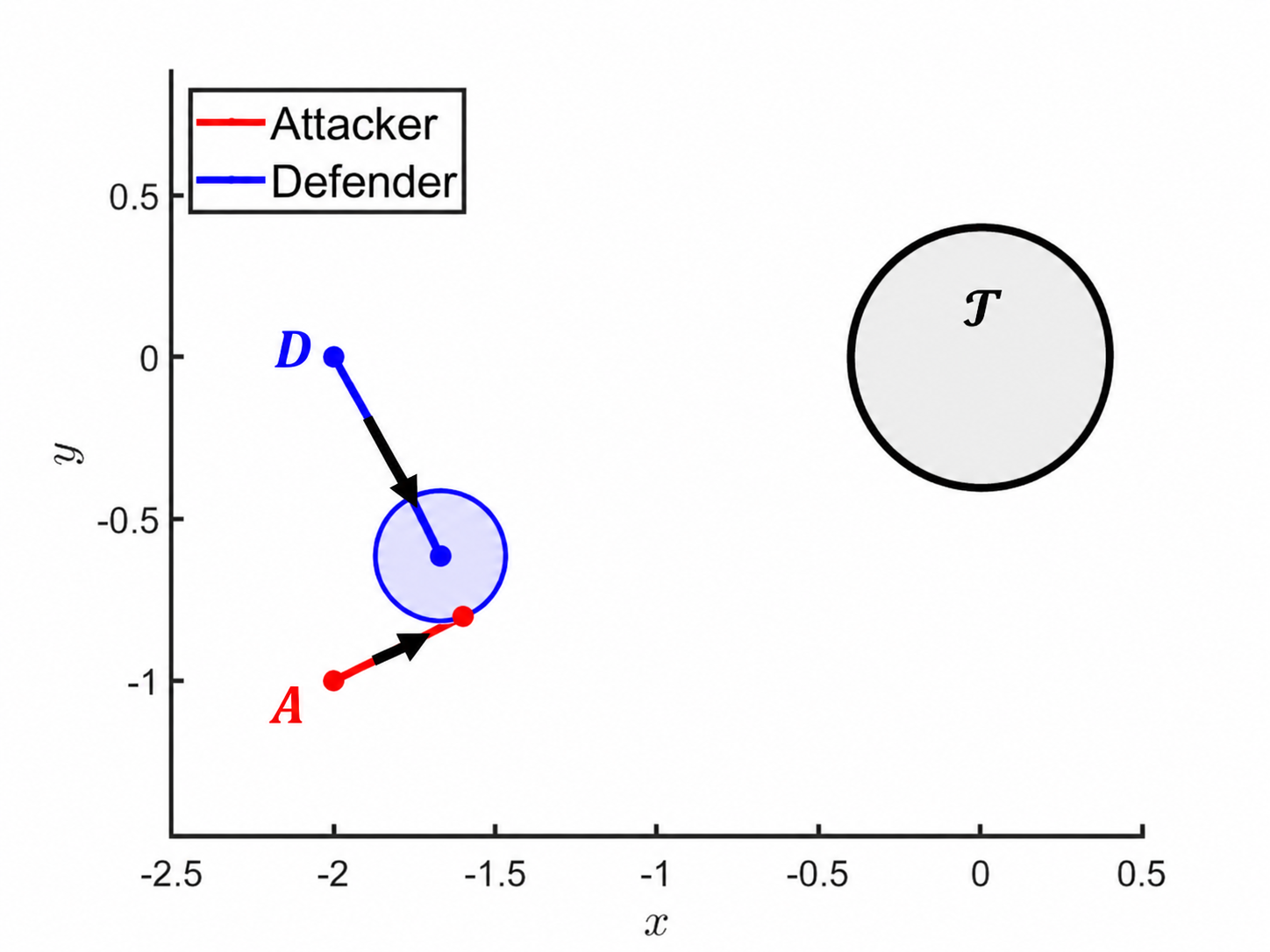}}
\subfigure[$A$ uses the optimal strategy while $D$ uses the pure-defense strategy.]{\label{fig1c}\includegraphics[width=2.3in]{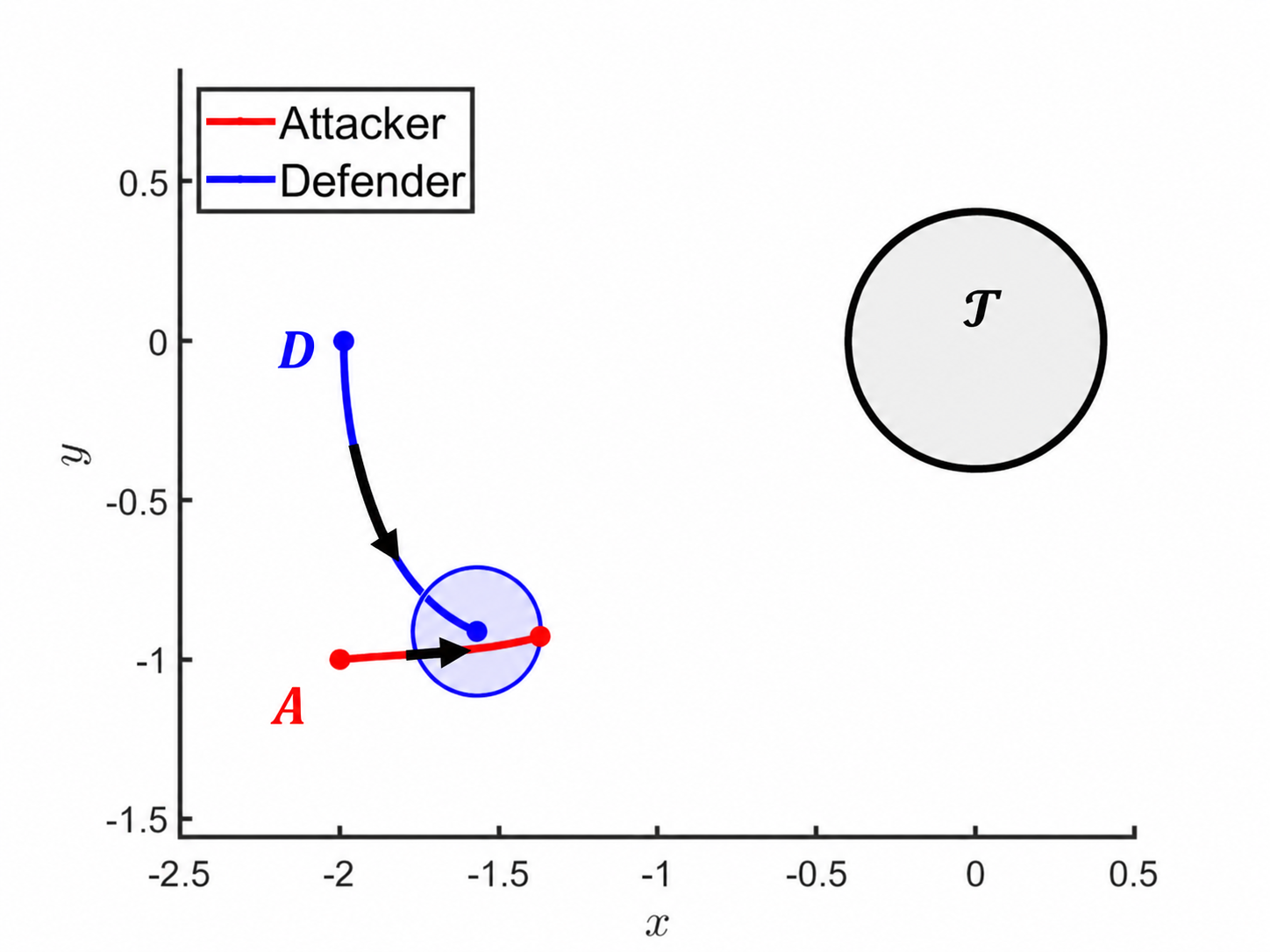}}
\caption{Simulation of the one-attacker-one-defender reach-avoid game with information delay. The red trajectory represents the attacker’s motion, the blue trajectory represents the defender’s motion, the blue circle denotes the defender’s capture region, and the black region denotes the target region.}
\label{fig1}
\end{figure*}
\begin{figure}[!t]
\centerline{\includegraphics[width=\columnwidth]{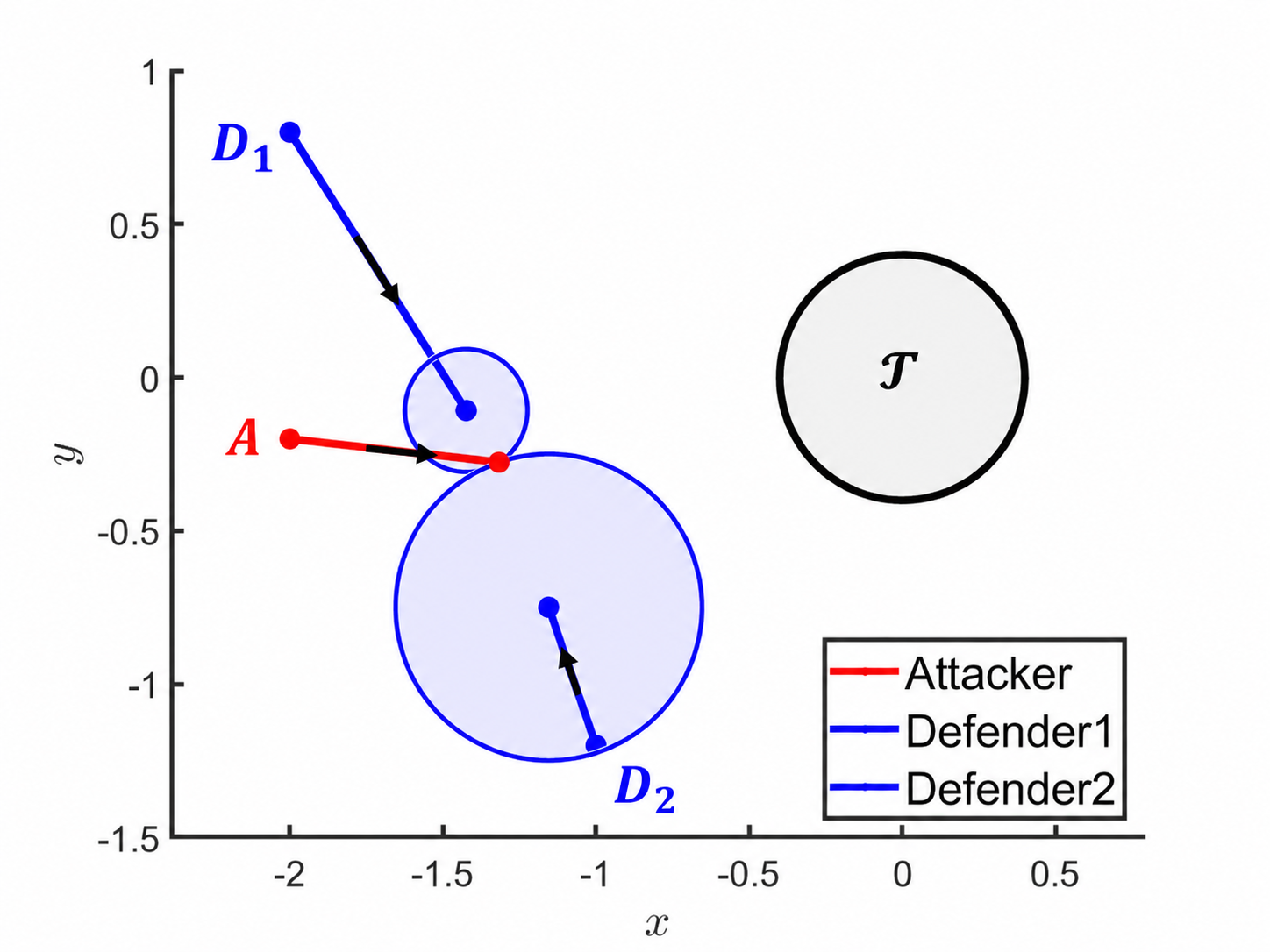}}
\caption{Simulation of the one-attacker-two-defender reach-avoid game with information delay. $A$ and the defenders use the optimal strategies. The red trajectory represents the attacker’s motion, the blue trajectories represent the defenders' motions, the blue circles denote the defenders' capture regions, and the black region denotes the target region.}
\label{fig2}
\end{figure}
\begin{figure}[!t]
\centerline{\includegraphics[width=\columnwidth]{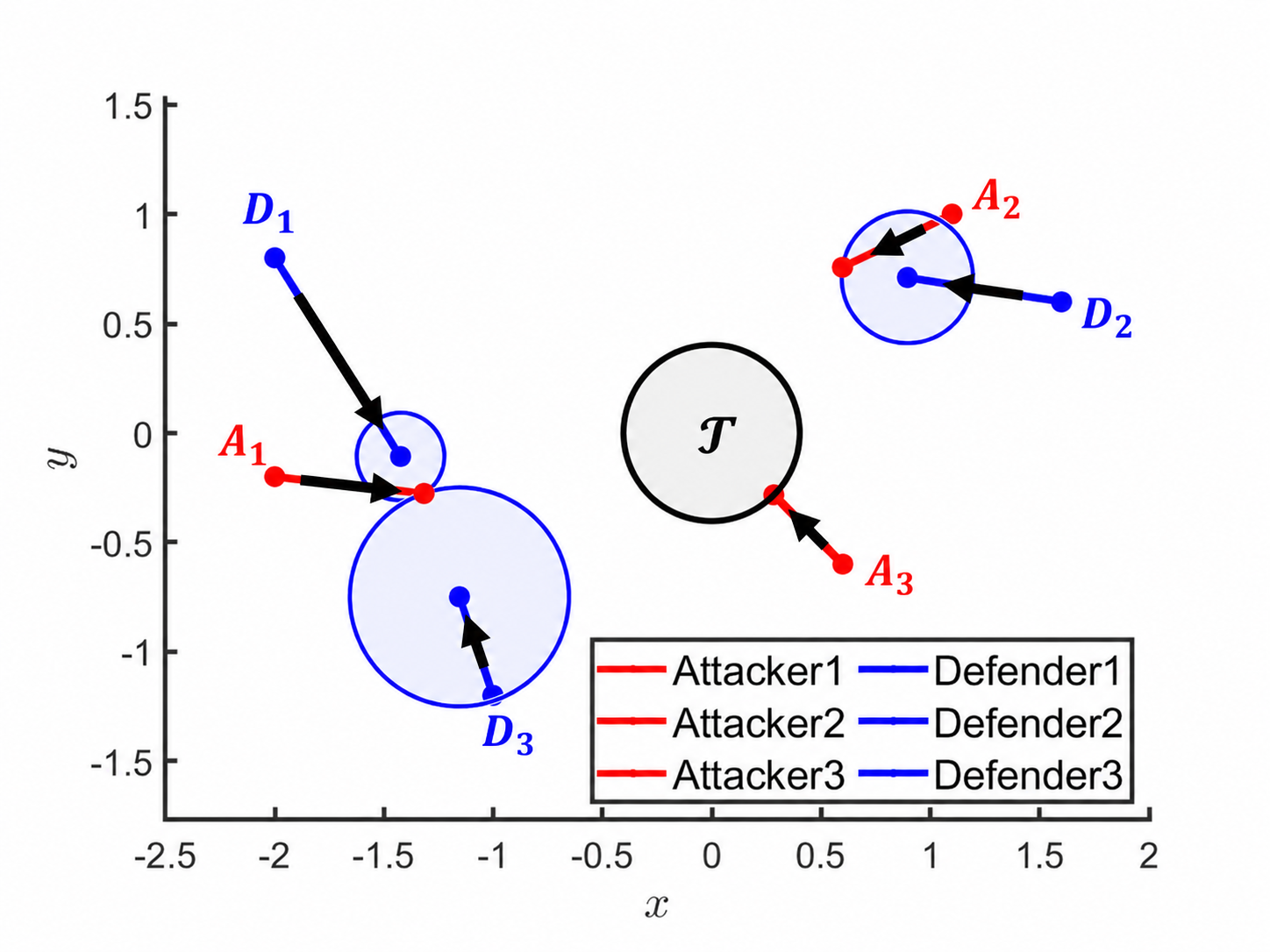}}
\caption{Simulation of the three-attacker-three-defender reach-avoid game with information delay. The attackers and defenders use the optimal strategies. The red trajectories represent the attackers' motions, the blue trajectories represent the defenders' motions, the blue circles denote the defenders' capture regions, and the black region denotes the target region.}
\label{fig3}
\end{figure}
First, we present simulation results for strategies ~\eqref{eq:optimalstrategyaftertau} and~\eqref{eq:optimalstrategybeforetau} in the one-attacker-one-defender reach-avoid game with information delay. In this simulation, the initial state is $\mathbf x_0=(-2,0,-2,-1)^\top$, with parameters $\bar v_D=1$, $\bar v_A=0.5$, $\tau=0.2$, $r=0.2$, $f(\mathbf x)=\|\mathbf x\|-0.4$. The corresponding trajectory is shown in Fig.~\ref{fig1a}, and the resulting cost is $J^*=1.3370$. To further demonstrate the effectiveness of these strategies, we compare them with two benchmark policies: a pure-attack strategy and a pure-defense strategy. In Fig.~\ref{fig1b}, the attacker adopts the pure-attack strategy, always moving toward the point in the target region closest to its current position. This yields a cost of $J=1.3900>J^*$, showing that unilateral deviation by the attacker leads to a worse outcome. In Fig.~\ref{fig1c}, the defender adopts the pure-defense strategy, always moving toward the attacker’s currently observed position. This yields a cost of $J=1.2568<J^*$, showing that unilateral deviation by the defender also results in a worse outcome for itself. These comparative results confirm the superiority of strategies~\eqref{eq:optimalstrategyaftertau} and~\eqref{eq:optimalstrategybeforetau}.

\subsection{Multi-Agent Case}
Next, we present the simulation results for strategies~\eqref{eq:optimalstrategyfora} and~\eqref{eq:optimalstrategyford} in the one-attacker-multiple-defender reach-avoid game with information delay. We consider a scenario involving two defenders, $D_1$ and $D_2$, and one attacker $A$. The initial positions are given by $\mathbf x_{D_1}^0=(-2,0.8)$, $\mathbf x_{D_2}^0=(-1,-1.2)$, and $\mathbf x_A^0=(-2,-0.2)$. Moreover, $\bar v_{D_1}=\bar v_{D_2}=1$, $\bar v_A=0.5$, $\tau_1=0.3$, $\tau_2=0.9$, $r_1=0.2$, $r_2=0.5$, $f(\mathbf x)=\|\mathbf x\|-0.4$. The simulation results are shown in Fig.~\ref{fig2}, with the corresponding cost value $J^*=0.9454$. For comparison, we also simulate the one-attacker-one-defender reach-avoid game with information delay under the same settings. The results show that the attack-defense cost between $A$ and $D_1$ is $J_1=0.9363<J^*$, and that between $A$ and $D_2$ is $J_2=0.7999<J^*$. These results indicate that,  in some scenarios, cooperation among multiple defenders can achieve better performance than a single defender acting alone.

Finally, we present simulation results for the $N$-$M$ matching problem. In this experiment, we consider three defenders $D_1$, $D_2$, $D_3$ and three attackers $A_1$, $A_2$, $A_3$. The initial positions of $D_1$, $D_2$ and $D_3$ are $(-2,0.8)$, $(1.6,0.6)$, and $(-1,-1.2)$, respectively, while those of $A_1$, $A_2$ and $A_3$ are $(-2,-0.2)$, $(1.1,1)$, and $(0.6,-0.6)$, respectively. Moreover, $\bar v_{D_1}=\bar v_{D_2}=\bar v_{D_3}=1$, $\bar v_{A_1}=\bar v_{A_2}=\bar v_{A_3}=0.5$, $\tau_1=0.3$, $\tau_2=0.4$, $\tau_3=0.9$, $r_1=0.2$, $r_2=0.3$, $r_3=0.5$, $f(\mathbf x)=\|\mathbf x\|-0.4$. The simulation results are shown in Fig.~\ref{fig3}. In this scenario, defenders $D_1$ and $D_3$ cooperate to capture attacker $A_1$, defender $D_2$ captures attacker $A_2$ on its own, and attacker $A_3$, for which no defender can be assigned, succeeds in entering the target region.

\section{Conclusion}\label{cc}
This paper studied multiplayer reach-avoid differential games with defender-side information delay. For the one-attacker-one-defender case, we derived an explicit analytical characterization of the delayed attack region and proved its convexity, which enables efficient determination of the winning condition. In the case where the defender can guarantee capture, we formulated a convex optimization problem to compute the capture point and obtained equilibrium strategies in the sense of a subgame-perfect Nash equilibrium. We further extended the framework to the one-attacker-multiple-defender case and the general multiplayer setting, where delay-aware winning relations were incorporated into a maximum matching formulation for defender-attacker assignment. Numerical simulations in one-vs-one, one-vs-two, and multi-agent cases validated the theoretical results and demonstrated the impact of optimal and non-optimal strategies under information delay. Future research  will focus on extending the framework to three-dimensional environments, more general dynamics, and distributed implementations with heterogeneous delays.

\section*{References}

\bibliographystyle{unsrt}
\bibliography{refs}

\end{document}